\documentclass[11pt, letterpaper]{article}
\usepackage{fullpage}

\usepackage{amsmath,amssymb,amsfonts}
\usepackage{epsfig}
\usepackage{latexsym,nicefrac,bbm}
\usepackage{xspace}
\usepackage[hidelinks]{hyperref} 
\usepackage[framemethod=tikz]{mdframed}
\usepackage[utf8]{inputenc}
\usepackage{thmtools,thm-restate}

\newcommand{\eps}{\varepsilon}
\newcommand{\SynAppx}{\text{\sc{SynAppx}}}
\newcommand{\res}{\mathrm{Syn}}

\newcommand{\val}{\mathrm{Val}}
\newcommand{\PCP}{\mathrm{PCP}}
\newcommand{\FPCP}{\mathrm{FPCP}}
\newcommand{\fsat}{\mathrm{Fsat}}
\newcommand{\poly}{\mathrm{poly}}
\newcommand{\pr}{\mathrm{Pr}}
\newcommand\B{\{0,1\}}      % boolean alphabet  use in math mode
\newcommand\N{\mathbb N}
\newcommand{\PTIME}{\mathrm{P}}
\newcommand{\NP}{\mathrm{NP}}
\newcommand{\DP}{\mathrm{DP}}
\newcommand{\coNP}{\mathrm{coNP}}
\newcommand{\ZPP}{\mathrm{ZPP}}

\newcommand{\inparen}[1]{\left(#1\right)}             %\inparen{x+y}  is (x+y)

\definecolor{Mygray}{gray}{0.8}

\newtheorem{theorem}{Theorem}[section]
\newtheorem{corollary}[theorem]{Corollary}
\newtheorem{lemma}[theorem]{Lemma}

\newtheorem{proposition}[theorem]{Proposition}
\newtheorem{definition}[theorem]{Definition}

\newtheorem{remark}[theorem]{Remark}

\newenvironment{proof}{\noindent{\bf Proof:}\hspace*{1em}}{\qed\bigskip}
\newcommand{\qed}{\hfill\ensuremath{\square}}

\usepackage{authblk}

\title{Strong inapproximability of the\\ shortest reset word\thanks{Supported by the \emph{NCN} grant 2011/01/D/ST6/07164.}}
\author[1]{Paweł Gawrychowski\thanks{Currently holding a post-doctoral position at Warsaw Center of Mathematics and Computer Science.}}
\author[2]{ Damian Straszak\thanks{Part of the work was carried out while the author was a student at Institute of Computer Science, University of Wrocław, Poland.}}
\affil[1]{Institute of Informatics, University of Warsaw, Poland}
\affil[2]{EPFL, Switzerland}

\begin{document}

\maketitle
\begin{abstract}

The Černý conjecture states that every $n$-state synchronizing automaton has a reset word of
length at most $(n-1)^2$. We study the hardness of finding short reset words. It is known that the exact version of the problem, i.e., finding
the shortest reset word, is $\NP$-hard and $\coNP$-hard, and complete for the $\DP$ class, and that approximating
the length of the shortest reset word within a factor of $O(\log n)$ is $\NP$-hard [Gerbush and Heeringa, CIAA'10], even for the binary alphabet [Berlinkov, DLT'13].
We significantly improve on these results by showing that, for every $\eps>0$, it is $\NP$-hard to approximate the length of the
shortest reset word within a factor of $n^{1-\eps}$. This is essentially tight since a simple $O(n)$-approximation algorithm exists.
\end{abstract}

\section{Introduction}\label{sec:introduction}
Let $A=(Q,\Sigma,\delta)$ be a deterministic finite automaton. We say that $w\in \Sigma^*$ resets (or synchronizes) $A$ if $|\delta(Q,w)|=1$,
meaning that the state of $A$ after reading $w$ does not depend on the choice of the starting state.
If at least one such $w$ exists, $A$ is called synchronizing. In 1964 Černý conjectured that every synchronizing $n$-state automaton admits a reset word of 
length $(n-1)^2$. The problem remains open as of today. It is known that an $\frac{n^3-n}{6}$ bound holds~\cite{Pin83} and that there are automata 
requiring words of length $(n-1)^2$. The conjecture was proved for various special classes of
automata~\cite{Ananichev2005,Eppstein1990,Grech13,Kari,Rystsov1997,Steinberg}.
%it is known that a random automaton admits a synchronizing word of length $n^{1+\eps}$ with high probability~\cite{Nicaud14}, which supports
%the conjecture.
For a thorough discussion of the Černý conjecture see~\cite{VolkovSurvey}.
%Quite a few generalizations and similar questions have been considered, the most famous one being the road coloring problem, finally solved by
%Trahtman~\cite{Trahtman08}. 

Computational problems related to synchronizing automata were also studied. It is known that finding the shortest reset word
is both $\NP$-hard and $\coNP$-hard~\cite{Eppstein1990}. Moreover, it was shown to be $\DP$-complete~\cite{Olschewski}. 

In this paper, rather than looking at the exact version, we consider the problem of finding {\emph short} reset words for automata, or to put it 
differently, the question of approximating the length of the shortest reset word. For a given $n$-state synchronizing automaton, we want to find
a reset word which is at most $\alpha$ times longer than the shortest one, where $\alpha$ can be either a constant or a function of $n$.
There is a simple polynomial time algorithm achieving $O(n)$-approximation~\cite{Gerbush11}.

\subsection{Previous work and our results.} 
Berlinkov showed that finding an $O(1)$-approximation is $\NP$-hard by giving a combinatorial reduction from SAT~\cite{BerlinkovConst}. Later, Gerbush and Heeringa~\cite{Gerbush11} used the $\log n-$approximation hardness of SetCover~\cite{Feige98}
to prove that $O(\log n)$-approximation of the shortest reset word is $\NP$-hard. Finally, Berlinkov~\cite{Berlinkov13} extended their result to hold
even for the binary alphabet, and conjectured that a polynomial time $O(\log n)$-approximation algorithm exists.
We refute the conjecture by showing that, for every constant $\eps>0$, no polynomial time $n^{1-\eps}-$approximation is possible 
unless $\PTIME=\NP$. This together with the simple $O(n)-$approximation algorithm gives a sharp threshold result for the shortest reset word problem.

The mathematical motivation and its algorithmic version considered in this paper are closely connected, although
not in a very formal sense. All known methods for proving bounds on the length of the shortest reset word are actually
based on explicitly computing a short reset word (in polynomial time). The best known method constructs
a reset word of length $\frac{n^3-n}{6}$, while the (most likely) true upper bound is just $(n-1)^2$, which is
smaller by a factor of roughly $\frac{n}{6}$.
Similarly, the best known (polynomial time) approximation algorithm achieves $O(n)$-approximation. Hence it is reasonable
to believe that an $o(n)$-approximation algorithm could be used to significantly improve the upper bound on the length of the
shortest synchronizing word to $o(n^3)$. In this context, our result
suggests that improving the bound on the length of the shortest synchronizing word to $O(n^{3-\eps})$ requires
non-constructive tools.

The main insight is to start with the PCP theorem. We recall the notion of
constraint satisfaction problems, and using the result of Håstad and Zuckerman provide a class
of hard instances of such problems with specific properties tailored to our particular application.
Then, we show how to appropriately translate such a problem into a synchronizing automaton.
%with either
%quite short synchronizing word or only long synchronizing words, depending on whether the
%initial instance is satisfiable or not.

%Conversely: approximation hardness of $\Omega(h(n))$ indicates in some sense that proving an $O(n^2 h(n))$ upper bound for the Černý problem is nontrivial (or at least giving a constructive -- algorithmic proof). In this context, our result suggests that proving an $O(n^{3-\eps})$ bound on the length of synchronizing words will require non-constructive tools, like pigeonhole principle with exponentially many pigeonholes or some non-constructive probabilistic arguments.

\subsection{Organization of the paper.} We provide the necessary definitions and the background on finite automata in the preliminaries.
We also introduce the notion of probabilistically checkable proofs and state the PCP theorem, then define constraint satisfaction problems
and their basic parameters.

In the next three sections we gradually move towards the main result. In Section~\ref{sec:constant} we prove that $(2-\eps)$-approximation of
the shortest reset word is $\NP$-hard. In Section~\ref{sec:power} we strengthen this by showing that, for a small fixed $\eps>0$,
$n^{\eps}-$approximation is also $\NP$-hard. Finally, in Section~\ref{sec:main}, we provide more background on probabilistically
checkable proofs and free bit complexity, and prove that, for every $\eps >0$, even $n^{1-\eps}$-approximation
is $\NP$-hard. Even though the final result subsumes
Sections~\ref{sec:constant}~and~\ref{sec:power}, this allows us to gradually
introduce the new components. 
 
In the Appendix we sketch how deduce the subconstant error PCP theorem from the classical version.
%All missing proofs can be found in the appendix.
%: the simple combinatorial reduction used for the $(2-\eps)-$hardness
%is then used to show the $n^{\eps}-$hardness, and then finally slightly modified
%for the $n^{1-\eps}-$hardness. All missing proofs can be found in the appendix.

\section{Preliminaries} \label{sec:preliminaries}
\subsection{DFA.}
A deterministic finite automaton (in short, an automaton) is a triple $A=(Q,\Sigma, \delta)$, where $Q$ is a nonempty finite set of states, $\Sigma$ is a nonempty finite alphabet, and $\delta$ is a transition function $\delta:Q\times \Sigma \to Q$. In the usual definition one includes additionally a starting state and a set of accepting states, which are irrelevant in our setting. Equivalently, we can treat an automaton as a collection of $|\Sigma|$ transformations of a finite set $Q$. We consider words over $\Sigma$, which are finite sequences of letters (elements of $\Sigma$). The empty word is denoted by $\eps$, the set of words of length $n$ by $\Sigma^n$, and the set
of all words by $\Sigma^{*}$. For $w\in \Sigma^*$, $|w|$ stands for the length of $w$ and $w_i$ is the $i$-th letter of $w$, for any $i\in\{1,2,\ldots,|w|\}$.

If $A=(Q,\Sigma, \delta)$ is an automaton, then we naturally extend $\delta$ from single letters to whole words by defining $\delta(q,\eps)=q$ and $\delta(q,wa)=\delta(\delta(q,w),a)$.
For $P\subseteq Q$ we denote by $\delta(P,w)$ the image of $P$ under $\delta(\cdot,w)$.

\subsection{Synchronizing Automata.}
An automaton $A=(Q,\Sigma, \delta)$ is synchronizing if there exists a word $w$ for which $|\delta(Q,w)|=1$. Such $w$ is then called a synchronizing (or reset) word and the length of a shortest such word
is denoted by $\res(A)$.
One can check if an automaton is synchronizing in polynomial time by verifying that every pair of states can be synchronized to a single state.

%Now the problem we want to study is as follows.

\vspace{1pt}
\begin{mdframed}
$\SynAppx(\Sigma,\alpha)$
\\
Given a synchronizing $n$-state automaton $A$ over an alphabet $\Sigma$, find a 
word of length at most $\alpha \cdot \res(A)$ synchronizing $A$.
Here both $\alpha$ and $|\Sigma|$ can be a function of $n$.
\end{mdframed}
\vspace{1pt} 

We are interested in solving $\SynAppx(\Sigma, \alpha)$ in polynomial time, with $\alpha$ as small
as possible.

\subsection{$O(n)-$Approximation.}\label{subsec:approximation}
It is known~\cite{Gerbush11} that for any fixed $k$ the problem $\SynAppx(\Sigma, \frac{n}{k})$ can be solved in $O(n^{k+1})$ time (we assume that $\Sigma$ is of constant size). The basic idea is that, for a given automaton $A=(Q,\Sigma, \delta)$, we construct a graph $G$ with the
vertex set $V=\{S\subseteq Q: |S|\leq k+1\}$ and a directed edge $S\rightarrow\delta(S,a)$ labeled with $a$ for every $S\in V$ and $a\in \Sigma$.
Then for a given $S\in V$ the shortest word synchronizing $S$ to a single state corresponds to the shortest path connecting $S$ to some singleton set
$\{q\}\in V$. Each such word is of length at most $\res(A)$. The algorithm works in $\left\lceil\frac{n}{k}\right\rceil$ phases. We start with the full set of states to reset $R:=Q$ 
and with an empty word $w:=\eps$, and in each phase we will decrease the size of $R$ by $k$, while assuring that $\delta(Q,w)=R$. In a single phase we take any 
subset $S$ of $R$ of size $k+1$ (if possible) and find the shortest word $w'$ resetting $S$ to a single state (note that $|w'|\leq \res(A)$). We set 
$w:=ww'$, $R:=\delta(R,w')$ and continue. One can easily see that in the end we obtain a synchronizing word $w$ of length at most
$\left\lceil\frac{n}{k}\right\rceil\cdot\res(A)$.

\subsection{Cubic Bound for the Černý Conjecture.}
Setting $k=1$ in the reasoning from Section~\ref{subsec:approximation}, we obtain an upper bound for $\res(A)$. This follows from the fact that the graph $G$ has $O(n^2)$ vertices and consequently every shortest path has length $O(n^2)$. In the end we have $\res(A)\leq |w|=(n-1)\cdot O(n^2)=O(n^3)$. In contrast, the famous Černý conjecture states that for every synchronizing automaton $A$ it holds $\res(A)\leq (n-1)^2$. Interestingly, the best bound known up to now is $\frac{n^3-n}{6}$~\cite{Pin83}, which is also cubic. Any $o(n^3)$ upper bound would be a very interesting result for this problem.

\subsection{Alphabet Size.} 
In the general case, the size of the
alphabet can be arbitrary. Our construction will use $\Sigma=\{0,1,2\}$, which
can be then reduced to the binary alphabet using the method of Berlinkov~\cite{Berlinkov13}.
It is based on encoding every letter in binary and adding some intermediate states. For completeness we state the appropriate lemma and sketch its proof.

\begin{lemma}[Lemma 7 of~\cite{Berlinkov13}]
\label{lemma:binarization}
Suppose $\SynAppx(\{0,1\},n^{\alpha})$ can be solved in polynomial time for some $\alpha \in (0,1)$, then so can be $\SynAppx(\Sigma,O(n^{\alpha}))$ for any $\Sigma$ of constant size.
\end{lemma}

\begin{proof}
As shown in Lemma 7 of~\cite{Berlinkov13}, given an $n$-state automaton $A$ over an alphabet $\Sigma$, one can efficiently construct
an automaton $B$ on $\tilde{n}:=2|\Sigma|n$ states over the binary alphabet, such that $\res(A)t \leq \res(B) \leq t(1+\res(A))$, where $t=\left\lceil \log_{2}|\Sigma|\right\rceil+1$. Then, if we can approximate $\res(B)$ within a factor of $\tilde{n}^{\alpha}$, we can compute in polynomial time an $x$ such that
$\res(B) \leq x \leq \res(B) \tilde{n}^{\alpha}$. Then $t\cdot \res(A)\leq x$ and $x\leq t(1+\res(A)) \tilde{n}^{\alpha}\leq 2t\res(A)\tilde{n}^{\alpha}$. Therefore, $\res(A) \leq \frac{x}{t} \leq 2\res(A)\tilde{n}^{\alpha}$, so $\frac{x}{t}$ approximates $\res(A)$ within a factor of
$2\tilde{n}^{\alpha}=O(n^\alpha)$.
\end{proof}

\subsection{PCP Theorems.}
We briefly introduce the notion of \textit{Probabilistically Checkable Proofs} (PCPs). For a comprehensive treatment refer to~\cite{Bellare98} or~\cite{AroraBarak}.

A polynomial-time probabilistic machine $V$ is called a $(p(n),r(n),q(n))$-PCP verifier for a language $L\subseteq \B^*$ if:
\begin{itemize}
\item for an input $x$ of length $n$, given random access to a ``proof'' $\pi \in \B^*$, $V$ uses at most $r(n)$ random bits, accesses at most $q(n)$ locations of $\pi$, and outputs $0$ or $1$ (meaning ``reject'' or ``accept'' respectively),
\item if $x\in L$ then there is a proof $\pi$, such that $\pr[V(x,\pi)=1]=1$,
\item if $x\notin L$ then for every proof $\pi$, $\pr[V(x,\pi)=1]\leq p(n)$.
\end{itemize}
We consider only \textit{nonadaptive} verifiers, meaning that the subsequently accessed locations
depend only on the input and the random bits, and not on the previous answers, hence
we can think that $V$ specifies at most $q(n)$ locations and then receives a sequence of
bits encoding all the answers.
$p(n)$ from the above definition is often called the \textit{soundness} or the \textit{error probability}. 
In some cases, also the \textit{proof length} is important. For a fixed input $x$ of length $n$ 
the proof length is the total number of distinct locations queried by $V$ over all possible $2^{r(n)}$ runs of $V$ (on different sequences of $r(n)$ random bits). The proof length is always at most
$q(n)\cdot 2^{r(n)}$, and such a bound is typically sufficient for applications, however in some cases we desire PCP-verifiers with smaller proof length.

The set of languages for which there exists a $(p,r,q)$-PCP verifier is denoted by $\PCP_{p}[r,q]$.
%One of the greatest achievements of computational complexity is the following characterization of $\NP.$
\begin{theorem}[PCP Theorem~\cite{AroraPV,AroraPCP98}]\label{thm:pcp}
$\NP=\PCP_{1/2}[O(\log n), O(1)]$.
\end{theorem}

\subsection{Constraint Satisfaction Problems.}
We consider \textit{Constraint Satisfaction Problems} (CSPs) over boolean variables. An instance of a general CSP over $N$ boolean variables $x_1, x_2, \ldots, x_N$ is a collection of $M$ boolean constraints $\phi=(C_1, C_2, \ldots, C_M)$, where a boolean constraint is just a function $C:\B ^N \to \B$. A boolean assignment $v:\B^N\to \B$ satisfies a constraint $C$ if $C(v)=1$,
and $\phi$ is satisfiable if there exists an assignment $v:\B^N\to \B$ such that $C_i(v)=1$ for all $i=1,2,\ldots, M$. We define $\val(\phi)$ to be the maximum fraction of constraints in $\phi$ which can be satisfied by a single assignment. In particular $\val(\phi) = 1$ iff $\phi$ is satisfiable.

We consider computational properties of CSPs. We are mainly interested in CSPs, where every $N$-variable
constraint has description of size $\poly(N)$ (as opposed to the naive representation using $2^N$ bits).
A natural class of such CSPs are CNF-formulas, where every constraint is a clause being a disjunction of
$N$ literals, thus described in $O(N)$ space. Another important class are qCSPs, where every constraint
\emph{depends} only on at most $q$ variables. Such a constraint can be described using $\poly(N, 2^q)$ 
space, which is polynomial whenever $q=O(\log N)$. Formally, we say that a clause $C$ depends
on variable $x_i$ if there exists an assignment $v\in \{0,1\}^N$ such that $C(v)$ changes after
modifying the value of $x_i$ and keeping the remaining variables intact.
We define $V_C$ to be the set of all such variables.
It is easy to see that $\phi(C)$ is determined as soon as we assign the values to all variables
in $V_C$.
Finally, the following class will be of interest to us.

\begin{definition}
Let $C$ be  an $N$-variable constraint and let $V_C$ be the set of variables on which $C$ depends.
Consider all $2^{|V_C|}$ assignments
$\B^{V_C} \to \B$. If only $K$ of such assignments satisfy $C$, we write $\fsat(C)\leq K$. $\fsat(\phi)\leq K$ if $\fsat(C)\leq K$ for every constraint $C$ in $\phi$.
\end{definition}

 According to the above definition, if $\phi$ is a qCSP instance then $\fsat(\phi) \leq 2^q$. 
A constraint $C$ such that $\fsat(C)\leq K$ can be described by its set $V_C$
and a list of at most $K$ assignments to the variables in $V_C$ satisfying $C$. Thus the description is polynomial in $N$ and $K$.
We will consider CSPs $\phi$ with $\fsat(\phi)\leq \poly(N)$ and always assume that they are represented as just
described.

\section{Simple Hardness Result}\label{sec:constant}

We start with a simple introductory result, which is that
for any fixed constant $\eps>0$,
it is $\NP$-hard to find for a given $n$-state synchronizing automaton $A$ a synchronizing word $w$
such that $|w|\leq (2-\eps)\cdot\res(A).$ The final goal is to prove a much strong result, but the basic construction
presented in this section is the core idea further developed in the subsequent sections. The construction
is not the simplest possible, nor the most efficient in the number of states of the resulting automaton, but
it provides good intuitions for the further proofs. For a simpler construction in this spirit see~\cite{BerlinkovConst}.

\begin{theorem}\label{thm:constant_easy}
For every constant $\eps>0$, $\SynAppx(\{0,1,2\}, 2-\eps)$ is not solvable
in polynomial time, unless $\PTIME=\NP$.
\end{theorem}

\subsection{Idea.}
Fix $\eps>0$. We will reduce 3-SAT to our problem, that is, show that an algorithm solving $\SynAppx(\{0,1,2\}, 2-\eps)$ 
can be used to decide satisfiability of 3-CNF formulas.
This will stem from the following reduction. For a given $N$-variable 3-CNF formula $\phi$ consisting of $M$ clauses
we can build in polynomial time a synchronizing automaton $A_\phi$ such that:
\begin{enumerate}
\item if $\phi$ is satisfiable then $\res(A_\phi)\approx N$,
\item if $\phi$ is not satisfiable then $\res(A_\phi)\geq 2N$.
\end{enumerate}
This implies Theorem~\ref{thm:constant_easy}, since applying an $(2-\eps)$-approximation algorithm to $A_\phi$
allows us to find out whether $\phi$ is satisfiable or not.

\subsection{Construction.}
Let $\phi=C_1\wedge C_2 \wedge \ldots \wedge C_M$ be a 3-CNF formula with $N$ variables $x_1,x_2,\ldots,x_N$ and $M$ clauses. We want to build an 
automaton $A_\phi=(\{0,1,2\},Q,\delta)$ with properties as described above. $A_\phi$ consists of $M$ gadgets, one for each clause in $\phi$, and a 
single sink state $s$. All letters leave $s$ intact, that is,
$\delta(s,0)=\delta(s,1)=\delta(s,2)=s$. We describe now a gadget for a fixed clause $C$. 

The gadget built for a clause $C$ can be essentially seen as a tree with $8$ leaves. Each leaf corresponds to one of the 
assignments to 3 variables appearing in $C$. First we introduce the uncompressed version of the gadget. Take all possible $2^N$ 
assignments and form a full binary tree of height $N$. Every edge in the tree is directed from a parent to its child and has a 
label from $\{0,1\}$. Every assignment naturally corresponds to a leaf in the tree. We could potentially use such a tree as
the gadget, except that its size is exponential. We will fix this by merging isomorphic subtrees to obtain a tree of size linear
in $N$. 

Let us denote by $L_0, L_1,\ldots, L_N$ the vertices at levels $0, 1,\ldots, N$, respectively, so that $L_0=\{r\}$, where $r$ is
the root, and $L_N$ is the set of leaves. 

Suppose that the variable $x_k$ does not occur in $C$. Take any vertex $v\in L_{k-1}$ and denote the subtrees rooted at
its children by $T_0$ and $T_1$. It is easy to see that $T_0$ and $T_1$ are isomorphic and can be merged, so that
we have two edges outgoing from $v$, labeled by $0$ and $1$, respectively, and both leading to the same vertex $v'$, which is
the root of $T_0$. We continue the merging until there are no more such vertices, which can be seen as
``compressing'' the tree. %We formalize the definition of a gadget in Appendix~\ref{subsec:formal_gadget}.
%The last layer $L_N$ of the resulting compressed tree
%contains states of the form $q_N^w$, where $w\in \B^3$ is some
%boolean assignment, and $\delta(q_N^w,0)=\delta(q_N^w,1)=s$ if $w $ satisfies $C$ and
%$\delta(q_N^w,0)=\delta(q_N^w,1)=r$ otherwise.

Let us now formalize the construction. The set of vertices at level $j$ (where $0\leq j\leq N$) is $L_j=\{q_{j}^w: w\in \{0,1\}^d\}$, where $d$ is the number of variables $x_i$ occurring in $C$ with $i\leq j$. For example $L_0$ is simply $\{ q_{0}^{\eps} \}$. Given $q_j^w$, one should think of $w\in \{0,1\}^d$ as some boolean assignment to variables $x_{j_1}, x_{j_2},\ldots, x_{j_d}$ appearing in $C$ (where $j_1\leq  j_2\leq  \ldots \leq j_d \leq j$). Let us now describe the edges. Take any variable $x_k$ and a vertex $q_{k-1}^w\in L_{k-1}$, then:
\begin{itemize}
\item if $x_k$ occurs in $C$, then $q_{k-1}^w$ has two distinct children $\delta(q_{k-1}^{w},0)=q_k^{w0}$ and $\delta(q_{k-1}^{w},1)=q_k^{w1}$,
\item if $x_k$ does not occur in $C$, then $q_{k-1}^w$ has one child $\delta(q_{k-1}^{w},0)=\delta(q_{k-1}^{w},1)=q_k^{w}$.
\end{itemize} 
We have already defined the edges outgoing from levels $0,1,\ldots, N-1$. This justifies the name
tree-gadget. It remains to define level $N$, where intuitively the ``synchronization'' or ``rejection'' happens. Let $q_N^w$ be a vertex on the last level, then:
\begin{itemize}
\item if $w$ corresponds to a satisfying assignment of $C$ then $\delta(q_N^{w},0)=\delta(q_N^{w},1)=s$,
\item otherwise $\delta(q_N^{w},0)=\delta(q_N^{w},1)=q_0^\eps.$
\end{itemize}

The above defined \textit{tree-gadget} will be further denoted by $T_C$, and its root $q_0^\eps$ will be usually referred
to as $r$. To complete the definition, we set $\delta(q,2)=r$ for every $q\in T_C$.
See Figure~\ref{figure:gadget_advanced} for an example.

\begin{figure}[t]
\centering
\includegraphics[width=\textwidth]{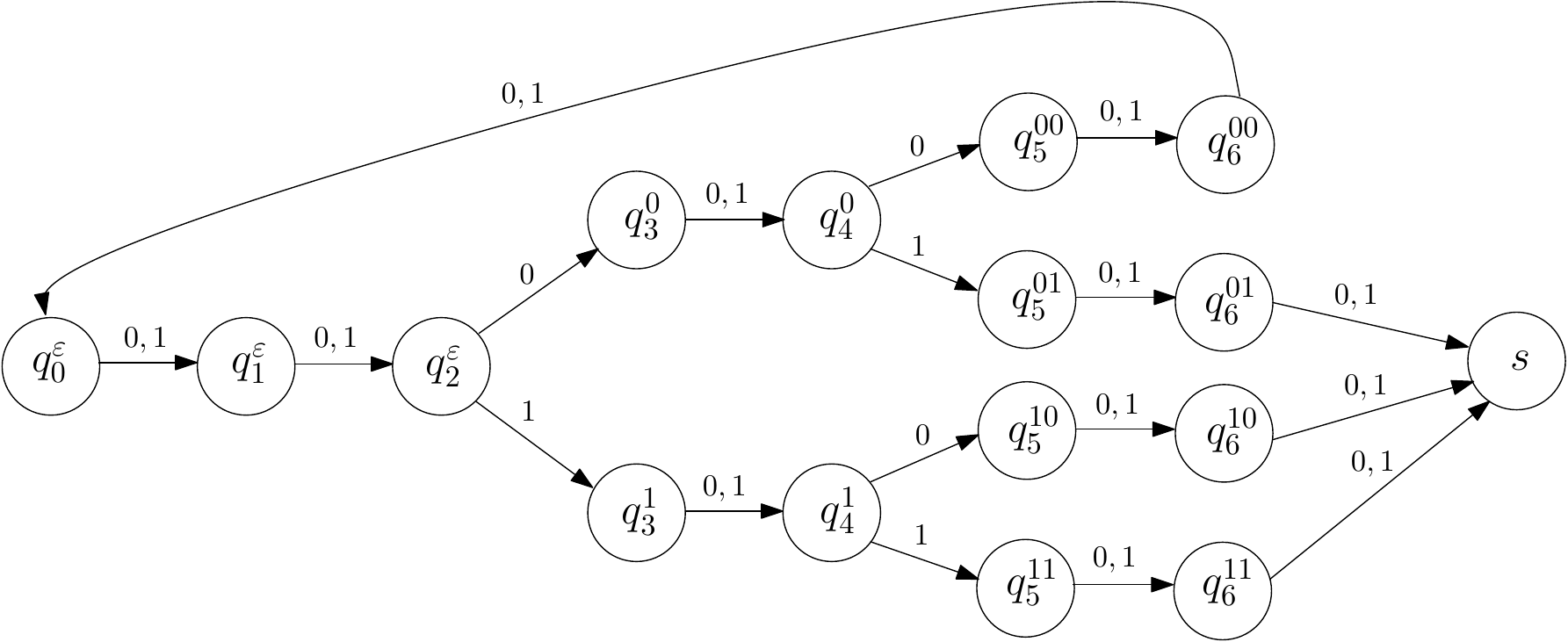}
\caption{Tree-gadget $T_C$ constructed for a clause $C=x_3\vee x_5$ and $N=6$.}
\label{figure:gadget_advanced}
\end{figure}

The automaton $A_\phi$ consists of $M$ disjoint tree-gadgets $T_{C_1}, T_{C_2}, \ldots, T_{C_M}$ and a single ``sink state'' 
$s$. Formally, its set of states is $Q=\sum_{i=1}^M T_{C_i} \cup \{s\}$ and the transition $\delta$ is defined 
above for every tree-gadget and the sink state $s$.

\subsection{Properties of $A_\phi.$} The following properties of $A_\phi$ can be established.

\begin{proposition}
\label{prop:assign}
Consider a tree-gadget $T_C$ with root $r$ constructed for a clause $C$.
%TODO nie definiujemy co to znaczy depends?
If $C$ depends only on variables $x_{j_1}, x_{j_2}, x_{j_3}$ then for any
binary assignment $v\in \B^N$ we have
$\delta(r,v) = q_N^{w}$, where $w=v_{j_1}v_{j_2}v_{j_3}$.
\end{proposition}
\begin{proof}
This follows immediately from the construction of $T_C$. We start in state $q_0^\eps$. Whenever we meet a relevant variable, we concatenate the assigned bit to our ``memory''. Thus after reading the whole assignment, we end up with the $3$ relevant bits. 
\end{proof}

Proposition~\ref{prop:assign} immediately yields the following.

\begin{corollary}\label{corr:gadget_synchro}
Consider a tree-gadget $T_C$ constructed for a clause $C$, and
let $w=vc$ be a binary word with $|v|=N$ and $c\in \{0,1\}$. If $v$ is an 
assignment satisfying $C$ then $\delta(r,w)=s$, otherwise $\delta(r,w)=r$.
\end{corollary}

Since $s$ is a sink state, synchronizing $A_\phi$ is equivalent to pushing all
of its states into $s$. Actually, it is enough to consider how to synchronize
the set $R=\{r_1,r_2,\ldots,r_M\}$, where $r_i$ is the root of the $i$-th tree-gadget
$T_{C_{i}}$. This is because $\delta(T_{C_i},2)=\{r_i\}$ for every $i$, hence one application
of letter $2$ ``synchronizes'' every gadget to its root and then it is enough to
synchronize the roots.

It is already easy to see that $A_\phi$ is always synchronizing, because we can synchronize gadgets one by one. The 
following lemma says that in case when $\phi$ is satisfiable, we can synchronize $A_\phi$ very quickly.

\begin{lemma}
\label{lemma:synchro_sat}
If $\phi$ is satisfiable and $v\in \{0,1\}^N$ is a satisfying assignment, then the word $w=2v0$ synchronizes $A_\phi$. Therefore $\res(A_\phi)\leq N+2$. 
\end{lemma}

\begin{proof}
Applying $2$ to the states $Q$ of $A_\phi$ yields $\delta(Q,2)=R\cup \{s\}$. Then, for every $r_i\in R$, since $C_i$ is satisfiable by $v$, we have by Corollary~\ref{corr:gadget_synchro} $\delta(r_i, v0)=s$. Hence $\delta(Q,w)=\{s\}$.
\end{proof}

Our next goal is to show that whenever $\phi$ is not satisfiable, $A_\phi$ cannot by synchronized quickly. To this end, we prove the following.

\begin{lemma}
\label{lemma:synchro_unsat}
Suppose there exists a binary word of length less than $2N+2$ synchronizing $R$ to $\{s\}$. Then $\phi$ is satisfiable.
\end{lemma}

\begin{proof}
First note that if for some tree-gadget $T_C$ with root $r$ and a word $v\in \B^{N+1}$ we have $\delta(r,v)=r$, then for every binary word $v'$ of length $d\leq N$ we have $\delta(r,vv')\in L_d$ (where $L_d$ is the $d$-th level of the tree-gadget), in particular $vv'$ does not push $r$ to $s$. This implies that if there is a word $w$ of length less than $2N+2$, synchronizing $A_\phi$, then there exists a word of length $N+1$ synchronizing $A_\phi$, because we can simply truncate it after the $(N+1)$-th letter. By Corollary~\ref{corr:gadget_synchro}, such a synchronizing word of length $N+1$ has then a prefix of length $N$, which is a satisfying assignment for $\phi$.
\end{proof}

By Lemma~\ref{lemma:synchro_unsat}, if there is no short binary word synchronizing $R$ then $\phi$
is not satisfiable. Using letter $2$ does not help at all in synchronizing $R$ as shown below.

\begin{lemma}
\label{lemma:synchro_unsat2}
Suppose there exists a word $w\in \{0,1,2\}^*$ synchronizing $R$ to $\{s\}$. Then there is a word $w'\in \B^*$ of length at most $|w|$ synchronizing $R$ to $\{s\}$. 
\end{lemma}

\begin{proof}
For convenience assume $w$ ends with a $2$. Decompose $w$ as follows $w=v_12v_22\ldots v_k2$, where $v_1, v_2,\ldots, v_k\in \B^*$. Fix one part of the form $v_i2$. If $|v_i|\leq N$, then for every $r\in R$, the empty word $\eps$ acts exactly the same on $r$ as $v_i 2$. Hence we can replace the part $v_i2$ by $\eps$ in w. More generally if $v_i$ is of the form $v_i=vv'$, where $v$ has length being a multiple of $N+1$ and $|v'|\leq N$, then the words $v_i2$ and $v$ are again equivalent with respect to action on $R$. Thus every part $v_i2$ can be replaced by a shorter binary word yielding a word $w'$ with the same action on $R$ as $w$. The lemma follows. 
\end{proof}

We are now ready to prove Theorem~\ref{thm:constant_easy}.

\begin{proof}[of Theorem~\ref{thm:constant_easy}]
Fix $\eps>0$ and suppose we can solve $\SynAppx(\{0,1,2\},2-\eps)$ in polynomial time.
We will show that we can solve 3-SAT in polynomial time. Let $\phi$ be any 3-CNF formula. We construct $A_\phi$ and approximate its shortest reset word within a factor of $2-\eps$.
By Lemma~\ref{lemma:synchro_sat} if $\phi$ is satisfiable then $\res(A_\phi)\leq N+2$, and if $\phi$ is not satisfiable then by Lemma~\ref{lemma:synchro_unsat} and Lemma~\ref{lemma:synchro_unsat2} we have $\res(A_\phi)\geq 2N+2$. Hence, $(2-\eps)$-approximation allows us to distinguish between those two cases in polynomial time.
\end{proof}

\section{Hardness with Ratio $n^{\eps}$}\label{sec:power}
In this section we show that it is possible to achieve a stronger hardness result using essentially the same reduction, but 
from a different problem. The problem we reduce from is CSP with some specific parameters. 
Its hardness is proved by suitably amplifying the error probability in the classical PCP theorem. We provide
the details below.

\subsection{PCP, qCSP and Probability Amplification.}
We want to obtain a hard boolean satisfaction problem, which allows us to perform more efficient reductions to  
$\SynAppx$. The usual source of such problems are PCP theorems, the most basic one asserting that
$\NP=\PCP_{1/2}[O(\log n), O(1)]$. By sequential repetition we can obtain verifiers erring with much lower probability, i.e.,
$\NP=\PCP_\eps[O(\log n),O(1)]$ for any fixed $\eps\in (0,1)$. Combining such a verifier with the construction of
$A_\phi$ described in the previous section yields that it is $\NP$-hard to approximate the shortest reset word within a factor
of $\alpha$, for any constant $\alpha$.
However, we aim for a stronger $n^{\eps}$-hardness for some $\eps>0$.
To this end we need to construct PCP verifiers with subconstant error.

Sequential repetition used to reduce the error probability as explained above has severe limitations.
We want the error probability to be $\approx n^{-1}$. This requires $\Theta(\log n)$ repetitions,
each consuming  fresh $O(\log n)$ random bits, and results in a verifier with the error probability
bounded by $n^{-1}$ using $O(\log n)$ queries.
The total number of used random bits is then $r=\Theta(\log^2 n)$, which is too much,
since the size of the automaton polynomially depends on $2^r$. 
Fortunately, the amount of used random bits can be reduced using the standard
idea of a random walk on an expander, resulting in the following theorem.
We explain the details and deduce the following theorem in Appendix~\ref{sec:expanders}
from Theorem~\ref{thm:pcp}.

\begin{restatable}[Subconstant Error PCP]{theorem}{restatablesubconst}
\label{thm:subconst_pcp}
$\NP \subseteq \PCP_{1/n}[O(\log n), O(\log n)]$.
\end{restatable}

Now we can use Theorem~\ref{thm:subconst_pcp} to prove the following.

\begin{theorem}
\label{thm:qcsp_red}
There exists a polynomial time reduction $f$, which takes a 3-CNF $n$-variable formula $\phi$
and returns a $qCSP$ instance $f(\phi)$ with $q=O(\log n)$, such that:
\begin{itemize}
\item if $\phi$ is satisfiable then $\val(f(\phi))=1$,
\item if $\phi$ is not satisfiable then $\val(f(\phi))\leq \frac{1}{n}$.
\end{itemize}
\end{theorem}

\begin{proof}
Take the PCP verifier $V$ for 3-SAT from Theorem~\ref{thm:subconst_pcp}. Assume it uses $r=O(\log n)$ random bits and queries the proof $q=O(\log n)$ times, for an $n$-variable formula $\phi$. One can see that the proof length $\ell$ is polynomial in $n$ (at most $q\cdot 2^r$). There will be $\ell$ variables in the resulting qCSP instance, one for every position in the proof. $f(\phi)$ consists of $2^r$ constraints, one for every possible sequence of random bits of length $r$. For a fixed sequence of random bits $s\in \B^r$, we create a constraint $C_s$. Given an assignment $v\in \B^l$, $C_s$ evaluates to $1$ if and only if $V$ accepts a proof $v$ (note that $C_s$ depends on at most $q$ variables). One can easily see that such a constraint satisfaction instance is satisfiable for satisfiable $\phi$. If $\phi$ is not satisfiable, then for every proof the probability of acceptance is at most $\frac{1}{n}$. It means that for every assignment $v\in \B^l$ at most $(\frac{1}{n})-$fraction of constraints can be satisfied by $v$. Finally, the reduction is polynomial time computable, because there are polynomially many constraints, each depending on $q=O(\log n)$ variables, thus one constraint can be described in $O(2^q)=\poly(n)$ time.
\end{proof}

\subsection{Construction.}
Let $\phi$ be an $N$-variables qCSP instance with $M$ clauses and $q=O(\log N)$. We want to
construct a synchronizing automaton $A_\phi$, such that the length of its shortest reset word allows us to reconstruct $\val(\phi)$ up to some error. 

The construction of $A_\phi$ is exactly the same as the one given for 3-CNF instances in Section~\ref{sec:constant}.
For a $q$-constraint $C$ we build a tree-gadget $T_C$ with $2^q$ leaves, each corresponding to an assignment to the variables $C$ depends on. As previously, the automaton has one sink state $s$ and $M$ tree-gadgets, one for every constraint. The construction still takes just polynomial time, the size of the automaton is polynomial in $M, N$ and $2^q=\poly(N)$.

\subsection{Properties of $A_\phi.$} Similarly as in the previous sections, the following
properties of $A_\phi$ can be established.
%in Appendix~\ref{sec:omitted}
%we prove the following properties of $A_\phi$.

\begin{lemma}
\label{lemma:qcsp_synchro_sat}
Let $\phi$ be a $N$-variable qCSP instance. If $\phi$ is satisfiable and $v\in \{0,1\}^N$ is a satisfying assignment, then the word $w=2v0$ synchronizes $A_\phi$. Therefore $\res(A_\phi)\leq N+2$. 
\end{lemma}

For the case when $\phi$ is not satisfiable we need a stronger statement than the one
from Lemma~\ref{lemma:synchro_unsat}. 

\begin{lemma}
\label{lemma:qcsp_synchro_unsat}
Let $\phi$ be a $N$-variable qCSP instance. If $w$ synchronizes $A_\phi$ then $|w|\geq \frac{1}{\val(\phi)}(N+1)$.
\end{lemma}

\begin{proof}
We prove a lower bound, thus we can focus on synchronizing a particular set of states. Let $R=\{r_1, r_2, \ldots, r_M\}$ be the set of roots of all tree-gadgets. Suppose $w$ synchronizes $R$ to $\{s\}$. By Lemma~\ref{lemma:synchro_unsat2} we can assume $w$ does not contain any occurrence of $2$. Also, we can assume (as in the proof of Lemma~\ref{lemma:synchro_unsat}) that the length of $w$ is a multiple of $(N+1)$. If it is not then we can cut out the last $|w| \mbox{ mod } (N+1)$ letters and the resulting word will still synchronize $R$ to $\{s\}$.

Decompose $w$ into the following parts: $w=v_1c_1v_2c_2\ldots v_kc_k$, where $v_i$ is a binary word of length $N$ and $c_i$ is a single binary character, for $i=1,2,\ldots, k$. We claim that for every constraint $C$ in $\phi$, some assignment $v_i$ (for $i\in \{1,2,\ldots, k\}$) satisfies $C$. Here a binary string of length $N$ is treated as a boolean assignment to the $N$ variables. Suppose for the sake of contradiction that there is a constraint $C$ in $\phi$ such that no $v_i$ satisfies $C$. Suppose $r$ is the root of the corresponding tree-gadget $T_C$. Using Corollary~\ref{corr:gadget_synchro} we can reason by induction that if $w'$ is a prefix of $w$ of length $d$ then $\delta(r,w')\in L_{d \; \mathrm{ mod } \; (N+1)}$. In particular $\delta(r,w)=r$.

Therefore we know that every constraint $C$ in $\phi$ is satisfied by some $v_i$. However, one assignment can satisfy at most $\val(\phi)$ constraints, hence $k\geq \frac{1}{\val(\phi)}$. The lemma follows.
\end{proof}

Now we are ready to prove the main theorem of this section.

\begin{theorem}\label{thm:power_easy}
There exists a constant $\eps>0$, such that $\SynAppx(\{0,1,2\}, n^\eps)$ is not solvable
in polynomial time, unless $\PTIME=\NP$.
\end{theorem}

\begin{proof}
We reduce 3-SAT to $\SynAppx(\{0,1,2\}, n^\eps)$, for some constant $\eps>0$. Let $\phi$ be an $n$-variable 3-CNF formula $\phi$. We use Theorem~\ref{thm:qcsp_red} to obtain
a qCSP instance $f(\phi)$ on $N$ variables and then convert it into a $G$-state automaton $A_{f(\phi)}$. If $\phi$ is 
satisfiable,
then by  Lemma~\ref{lemma:qcsp_synchro_sat} $\res(A_{f(\phi)})\leq N+2$. On the other hand, if $\phi$ is not satisfiable,
then $\val(f(\phi))\leq 1/n$, hence by Lemma~\ref{lemma:qcsp_synchro_unsat} $\res(A_{f(\phi)})\geq n(N+1)$. The
ratio between  those two quantities is $\frac{n(N+1)}{N+2}=\Omega(n)$. 

It remains to show that $n=\Omega(G^\eps)$ for some constant $\eps>0$.
In other words, we need to show that $G$ is polynomial in $n$. This holds,
because $f$ is a polynomial time reduction, hence $N,M = \poly(n)$ and the size of $A_{f(\phi)}$ is polynomial with respect to $N,M, 2^q$, but $q=O(\log N)=O(\log n)$ so $2^q=\poly(n)$.
\end{proof}

\begin{remark}
By keeping track of all the constants, one can obtain $n^\eps-$hardness for $\eps\approx 0.0095$,
but this is anyway subsumed by the next section.
\end{remark}

\section{Hardness with Ratio $n^{1-\eps}$}\label{sec:main}
In this section we prove the main result of the paper. 
It is not enough to use the reasoning from the previous section and simply optimize the constants. In fact, the 
strongest hardness result that we can possibly obtain by applying Theorem~\ref{thm:subconst_pcp} is $n^{\eps}$ for some 
tiny constant $\eps>0$. This stems from the fact that in our reduction we require the number of queries $q$ to be 
logarithmic in the size of the instance. If this is not the case, then the reduction takes superpolynomial time.
However, the crucial  observation is that the reduction can be modified so that we do not need the query complexity of
the verifier to be logarithmic. It suffices that the free bit complexity (defined below) is logarithmic.

\begin{theorem}\label{thm:power_hard}
For every constant $\eps>0$, $\SynAppx(\{0,1,2\}, n^{1-\eps})$ is not solvable
in polynomial time, unless $\PTIME=\NP$.
\end{theorem}

\subsection{Free Bit Complexity and Stronger PCP Theorems.}
Let us first briefly introduce the notion of free bit complexity.
For a comprehensive discussion see~\cite{Bellare98}.

\begin{definition}
Consider a PCP verifier $V$ using $r$ random bits on any input $x$.
For a fixed input $x$ and a sequence of random bits $R\in\B^r$,
define $G(x,R)$ to be the set of sequences of answers to the questions
asked by $V$, which result in an acceptance.
We say that $V$ has free bit complexity $f$ if $|G(x,R)|\leq 2^f$ and there is a polynomial time algorithm which
computes $G(x,R)$ for given $x$ and $R$.
\end{definition}

The set of languages for which there exists a verifier with soundness $p$, free bit complexity $f$,
and proof length $\ell$ is denoted by $\FPCP_p[r,f,\ell]$.

Håstad in his seminal work~\cite{Hastad96} proved that approximating the maximum clique within the factor $n^{1-\eps}$ 
is hard, for every $\eps>0$. To obtain this result he constructs PCP verifiers with arbitrarily small \textit{amortized free bit complexity}\footnote{Amortized free bit complexity is a parameter of a PCP verifier which essentially corresponds to the ratio between the free bit complexity and the logarithm of error probability.}. We state his result in a more recent and stronger version~\cite{HK01}:

\begin{theorem}\label{thm:amort_hastad}
For every $\eps>0$, there exist constants $t\in \N$ and $\alpha, \beta>0$ such that $\NP\subseteq \FPCP_{2^{-t}}[\beta \log n, \eps \cdot t, n^\alpha].$
\end{theorem}

In the next step we need to amplify the error probability, as we did in the previous section. It turns out that the
amplification using expander walks is too weak for our purpose. Håstad~\cite{Hastad96}, following the approach of
Bellare et al.~\cite{Bellare98}, uses sequential repetition together with a technique to reduce the demand for random bits. (See Proposition~11.2, Corollary~11.3 in~\cite{Bellare98}.) Unfortunately, this procedure involves randomization, so
his MaxClique hardness result holds under the assumption that $\ZPP\neq \NP$.

In his breakthrough paper Zuckerman~\cite{Zuckerman06} showed how to derandomize Håstad's MaxClique hardness result 
by giving a deterministic method for amplifying the error probability of PCP verifiers. He constructs very efficient 
randomness extractors, which then by known reductions allow to perform error amplification. One can conclude the 
following result from Theorem~\ref{thm:amort_hastad} and his result (see also Lemma~6.4 and
Theorem~1.1 in~\cite{Zuckerman06}). 

\begin{theorem}\label{thm:zuckerman}
For every $\eps>0$, there exist $c, \alpha>0$ such that for $t=c\log n$ it holds that
$\NP\subseteq  \FPCP_{2^{-t}}[(1+\eps) t, \eps t, n^\alpha]$.
\end{theorem}

Based on the above theorem, we can prove the following very strong analogue of Theorem~\ref{thm:qcsp_red}. \begin{theorem}
\label{thm:qcsp_free}
For every $\eps>0$, there exists a polynomial time reduction $f$, which takes an $n$-variable 3-CNF formula $\phi$
and returns an $N$-variable $CSP$ instance $f(\phi)$ with $M$ constraints, such that:
\begin{itemize}
\item $N\leq M^{\eps}$,
\item if $\phi$ is satisfiable then $\val(f(\phi))=1$,
\item if $\phi$ is not satisfiable then $\val(f(\phi))\leq \frac{1}{M^{1-\eps}}$,
\item $\fsat(f(\phi))\leq M^{\eps}$.
\end{itemize}
\end{theorem}

\begin{proof}
Fix an $\eps>0$. We use Theorem~\ref{thm:zuckerman} to conclude $\NP\subseteq  \FPCP_{2^{-t}}[(1+\eps) t, \eps t, n^\alpha]$ for some $\alpha, c>0$ and $t=c\log n$. As a preliminary step, we make the proof length negligible (because $\alpha$ can be arbitrarily big) using the simple sequential repetition. By repeating the verification procedure $k$ times we obtain $\NP\subseteq  \FPCP_{2^{-kt}}[(1+\eps) kt, \eps kt, n^\alpha]$. 
We take such verifier $V$ for the 3-SAT language with $k$ chosen large enough to make sure that $\alpha < \eps kc$.

Consider any 3-CNF formula $\phi$. For convenience denote $r:= (1+\eps)kt=(1+\eps) kc\log n$. We construct $f(\phi)$ as in the proof of  Theorem~\ref{thm:qcsp_red}. The number of variables is $N=n^\alpha$.
For every possible sequence $s\in \B^r$ of $r$ random bits we define a constraint $C_s$, such that
given an assignment $v\in \B^N$, $C_s$ evaluates to $1$ if and only if $V$ accepts the proof v. 

We have defined $f(\phi)$, now we prove that it satisfies all the claimed properties. $N=n^\alpha$,  $M=n^{(1+\eps)kc}$ and $\frac{\alpha}{kc}<\eps$, hence $N\leq M^{\eps}$. If $\phi$ is satisfiable then there is a proof which is accepted for all possible sequences of random bits, hence all the constraints can be satisfied simultaneously, so $\val(f(\phi))=1$. Suppose $\phi$ is not satisfiable, then the probability of accepting a wrong proof is at most $2^{-kt}=n^{-kc}=M^{-1/(1+\eps)}\leq \frac{1}{M^{1-\eps}}$.
Finally, because of the bound on the free bit complexity of $V$, for every sequence of random bits $s\in \B^r$ there are at most $2^{\eps kt}=n^{\eps kc} \leq M^{\eps}$ sequences of bits encoding
the answers to the queries, which result in an acceptance.
By the definition of free bit complexity, those sequences can be efficiently listed,
hence the CSP instance $f(\phi)$ can be  constructed in polynomial time.
\end{proof}

\subsection{Construction.}
Let $\phi$ be an $N$-variable CSP instance with $M$ constraints such that $\fsat(\phi)\leq K$ (for a
parameter $K$ to be chosen later). We want to construct an automaton $\hat{A}_\phi$
of size polynomial in $K$ and the size of $\phi$ such that $\res(\hat{A}_\phi)\approx \frac{N}{\val(\phi)}$.

Using $A_\phi$ as in the previous sections gives an automaton
of superpolynomial size, so we need to tweak it. Take any constraint $C$ and suppose it depends on $q$ variables.
Consider the tree-gadget $T_C$ built for $C$. We cannot assume that $q=O(\log n)$ as in the 
Section~\ref{sec:power}. In consequence, $T_C$ can be of exponential size, because the only possible
bound on its number of leaves is $2^q$. However, we have a bound $K$ on the number of essentially satisfying assignments. We will modify the definition 
of $T_C$, so that its size depends polynomially on $K$ rather than $2^q$.

Observe that in the original construction, at most $K$ out of $2^q$ leaves of $T_C$ correspond to satisfying assignments
(think of $K$ much smaller than $2^q$). Imagine for a moment a subtree of $T_C$ corresponding to the at most $K$ satisfying assignments.
The size of such subtree is at most $NK$, but it is not yet a good candidate for our gadget,
because some transitions are not well 
defined. However, this is not difficult to fix to obtain an equivalent \textit{compressed tree-gadget}
$\hat{T}_C$ as follows. Take a variable $x_k$ (such that $C$ depends on $x_k$) and a node $q_k^w$ with $w=vc$ (then basically $c\in \B$ corresponds to the boolean value assigned to $x_k$). Suppose that every leaf in the subtree rooted at $q_k^w$ corresponds to a non-satisfying assignment. Suppose further that there is some satisfying assignment in the subtree $q_{k-1}^{v}$ (then the leaf with a satisfying assignment must necessarily lie in the subtree rooted at $q_k^{v\bar{c}}$). In such a case there is nothing interesting happening in the subtree rooted at $q_k^w$, it has height $h=N-k$ and all of its leaves have two edges (labeled $0$ and $1$) going back to the root $r$ of $T_C$. We remove this subtree rooted at $q_k^w$ and instead attach a path of length $h$ to $q_k^w$, we also add transitions from the endpoint of the path to $r$. 

For an illustration of the compression procedure refer to Figure~\ref{figure:gadget_compressed}. The leftmost vertex in the figure is $q_{k-1}^v$, his two children are $q_k^{v0}$ and $q_k^{v1}$, the subtree rooted at $q_k^{v0}$ gets compressed. In a tree-gadget the black leaf in the picture has transitions (labeled by $0$ and $1$) to the sink vertex $s$ and all the grey vertices have transitions back to the root $r$. We obtain the compressed tree-gadget $\hat{T}_C$ by applying the above transformation to every relevant node $q_k^w$.

The automaton $\hat{A}_\phi$ is built analogously to $A_\phi$, with the crucial difference that we use $\hat{T}_C$
instead of $T_C$. One can see that $\hat{T}_C$ can be constructed in time polynomial in its size
by proceeding from the root down to the leaves. Furthermore, the resulting automaton
$\hat{A}_\phi$ is small.

\begin{figure}[t]
\centering
\includegraphics[width=\textwidth]{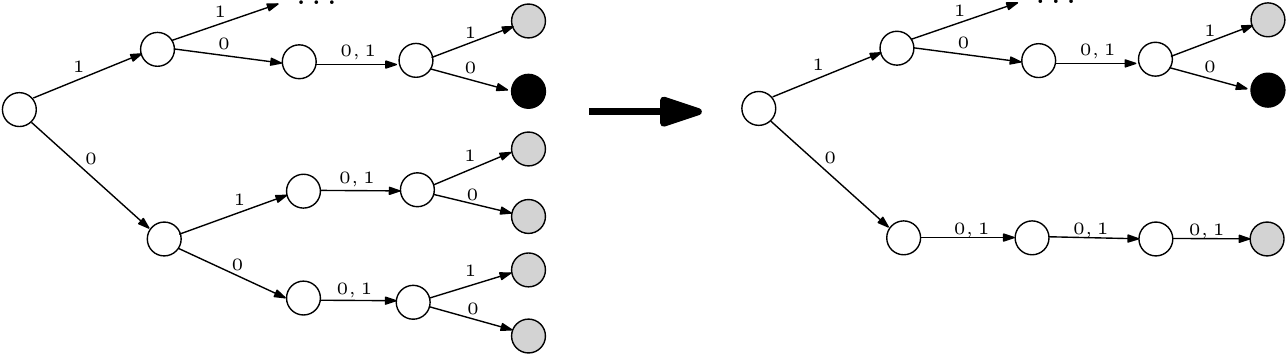}
\caption{Compressing the tree. The grey and black leaves correspond to non-satisfying and satisfying assignments,
respectively.}
\label{figure:gadget_compressed}
\end{figure}

\begin{lemma}
\label{thm:small_compressed}
Suppose $\phi$ is an $N$-variable CSP instance with $M$ constraints and $\fsat(\phi)\leq K$. Then the size of
$\hat{A}_\phi$ is $O(MN^2K)$.
\end{lemma}
\begin{proof}
The automaton consists of $M$ gadgets and $1$ additional state $s$. It suffices to prove that every gadget has size $O(N^2K)$. Fix any constraint $C$ in $\phi$ and consider the compressed tree-gadget $\hat{T}_C$. Suppose we remove from $\hat{T}_C$ all leaves corresponding to non-satisfying assignments together with the paths created in the compression procedure. What remains is a subtree $T_0$ of height $N$ with at most $K$ leaves, each corresponding to the satisfying assignments of $C$. The size of $T_0$ is $O(NK)$. To get $\hat{T}_C$ back from $T_0$, we need to attach paths (of length at most $N$) to some vertices of $T_0$ (at most one path per vertex). We add at most $O(NK)$ paths consisting of at most $N$ nodes each, hence the total size of $\hat{T}_C$ is $O(N^2K)$.
\end{proof}

\subsection{Properties of $\hat{A}_\phi$.}
The lemma below summarizes the properties of $\hat{A}_\phi$. Its proof is very similar to the
proofs of the lemmas summarizing the properties of $A_\phi$ and hence skipped.
%TODO which lemmas maybe?

\begin{lemma}
\label{thm:compr_aut}
Let $\phi$ be an $N$-variable CSP instance with $M$ constraints and $\fsat(\phi)\leq K$. Then $\hat{A}_\phi$ is a 
synchronizing automaton of size $O(MN^2K)$, which can be constructed in polynomial time. Furthermore, if $\phi$
is satisfiable then $\res(\hat{A}_\phi)\leq N+2$ and otherwise $\res(\hat{A}_\phi)\geq \frac{N+1}{\val(\phi)}$.
\end{lemma}

\begin{proof}[of Theorem~\ref{thm:power_hard}]
Fix any $\eps>0$. We reduce 3-SAT to $\SynAppx(\{0,1,2\}, n^\eps)$. Let $\phi$ be an $n$-variable 3-CNF formula. Then by Theorem~\ref{thm:qcsp_free} 
we can construct an $N$-variable CSP instance $f(\phi)$ with $M=\poly(n)$ constraints and $\fsat(f(\phi))\leq K= M^{\eps}$, 
where $N\leq M^\eps$. We know that if $\phi$ is satisfiable then $f(\phi)$ is satisfiable as well and if $\phi$ is not 
satisfiable then $\val(f(\phi))\leq \frac{1}{M^{1-\eps}}$. Then 
by Theorem~\ref{thm:compr_aut} we can construct $\hat{A}_{f(\phi)}$, which is an automaton of size 
$O(MKN^2)=O(M^{1+3\eps})$. If $\phi$ is satisfiable, $\res(\hat{A}_{f(\phi)})\leq N+2$ and if $\phi$ is not satisfiable
then $\res(\hat{A}_{f(\phi)})\geq \frac{N+1}{\val(f(\phi))}$. The ratio of those two bounds is:
$$\frac{N+1}{(N+2)\val(f(\phi))} = \Theta\inparen{\frac{1}{\val(f(\phi))}} = \Theta\inparen{M^{1-\eps}}$$
The size of the automaton $\hat{A}_{f(\phi)}$ is $G=O(M^{1+3\eps})$, so the above ratio can be related to the size of
the automaton as $\Omega\inparen{G^{\frac{1-\eps}{1+3\eps}}}=\Omega\inparen{G^{1-4\eps}}$. Hence assuming 
$\PTIME\neq \NP$, approximating the shortest reset word within ratio $G^{1-4\eps}$ in polynomial time is not possible. 
\end{proof}
\bibliographystyle{splncs03}
\bibliography{references}

\appendix

\section{Expanders and Error Reduction}\label{sec:expanders}
In this section we show how to deduce a subconstant error PCP theorem from the standard version, i.e., 
Theorem~\ref{thm:pcp}. Thus, we prove the following.

\restatablesubconst*

To this end we need a method of reusing random bits when repeating a random experiment introduced by Impagliazzo and Zuckerman~\cite{Recycle}. The idea is that, instead of generating fresh random bits for each repetitions of the experiment, we use a random walk on an expander to construct a pseudorandom sequence of bits, which is then used in subsequent repetitions. Because of expanding properties of such graphs, which we summarize below, this is enough to significantly decrease the error probability. A good reference for expander graphs and pseudorandom constructions is~\cite{AroraBarak}. For completeness we include all essential definitions, but we refer to the book for the proofs.

Whenever $G$ appears in the following text, it denotes an undirected $d$-regular graph on $n$ vertices, possibly containing loops and parallel edges.

\begin{definition}[$\lambda(G)$] Let $A_G$ be the random walk matrix of $G$ (that is, $A_G$ is the adjacency matrix of $G$ with each entry scaled by $\frac{1}{d}$). Let $\lambda_1, \lambda_2,\ldots, \lambda_n \in [-1,1]$ be the eigenvalues of $A_G$, sorted so that $|\lambda_1|\geq |\lambda_2|\geq \ldots \geq |\lambda_n|$. We define $\lambda(G)$ to be $|\lambda_2|$.
\end{definition}

%We proceed with the definition of an expander.
\begin{definition}[$(n,d,\lambda)$-expander] If $G$ is an $n$-vertex $d$-regular multigraph with $\lambda(G)\leq \lambda<1$, then we say that $G$ is an $(n,d,\lambda)$-graph.
\end{definition}

%\todo{motywacje dla expanderow}
It turns out that constructing a family of $(n,d,\lambda)$-graphs for some fixed $\lambda>0$ is a pretty simple task, since a random $d$-regular graph is 
an expander with high confidence. However, a true challenge is to construct expanders explicitly and without any use of random bits. A beautiful example 
of such a construction was given by Margulis~\cite{Margulis73}, its analysis was later improved and simplified first by Gabber and Galil~\cite{Galil}, 
and then by Jimbo and Maruoka~\cite{Jimbo85}, to yield the following.

\begin{theorem}\label{thm:margulis}
Let $G_{n^2}$ be the $8$-regular graph on vertex set $\mathbb{Z}_n^2$, with edges defined as follows: $(x,y)$ has neighbors
$(x\pm 2y,y), (x\pm (2y+1),y), (x,y\pm 2x), (x,y\pm (2x+1))$ (addition is performed modulo $n$). Then $G_{n^2}$ is an
$(n^2, 8, \frac{5\sqrt{2}}{8})$-graph.
\end{theorem}

\noindent The following theorem can be now used to reduce the error probability.

\begin{theorem}[Expander walks, 21.12 in \cite{AroraBarak}]\label{thm:expander_walks}
Let $G$ be an $(n,d,\lambda)$-graph and let $B\subseteq [n]$ be a set satisfying $|B|\leq \beta n$ for some $\beta \in (0,1)$. Let $X_1,X_2, \ldots, X_k$ be random 
variables denoting a $(k-1)$-step random walk in $G$, meaning that $X_1$ is chosen uniformly from $[n]$ and $X_{i+1}$ is a uniform random neighbor of 
$X_i$. Then $P(X_1 \in B \wedge X_2 \in B \wedge \ldots \wedge X_k \in B)\leq ((1-\lambda)\sqrt{\beta}+\lambda)^{k-1}$.
\end{theorem} 

\restatablesubconst*

\begin{proof}
Take any language $L\in \NP$, we would like to show that there exists a polynomial time $(\frac{1}{n}, O(\log n), O(\log n))-$verifier for $L$. By the basic Theorem~\ref{thm:pcp} we know that there is a $(\frac{1}{2}, O(\log n), O(1))-$verifier for $L$. Suppose it consumes at most $r\log n$ random bits. We intend to define another verifier $V'$, which makes more queries to the proof and needs 
more random bits, but its failure probability is at most $\frac{1}{n}$. Running $V$ independently $\Omega(\log n)$ times and checking
if there was at least one reject is not acceptable, as it increases the number of used random bits to $\Omega(\log^{2}n)$ bits.
For this reason, instead of making $\Omega(\log n)$ fully independent runs, we save some random bits using expanders. Let $k=c\log n$ be the number of repetitions. Fix an input $x$ of length $n$. We construct an 
$(\mathcal{O}(2^{r\log n}),8,\lambda)$-graph with $\lambda=\frac{5\sqrt{2}}{8}$ from Theorem~\ref{thm:margulis}, 
select a random starting vertex $X_1$ there, and choose a random walk of length $k-1$ starting from $X_{1}$
obtaining vertices $X_2, X_3, \ldots, X_k$. Now run $V$ $k$ times using $X_{i}$ as the required stream of $r\log n$ bits for the $i$-th run.
Answer $1$ if and only if all runs returned $1$. To obtain the sequence $X_1, \ldots, X_k$ we use $r$ random bits for $X_1$ and $(k-1)\cdot 3$ bits for $X_2, \ldots X_k$ (we need $3$ bits to pick one neighbor out of five), $O(\log n)$ random bits in total. Let us now calculate the error probability. According to Theorem~\ref{thm:expander_walks} (where we choose $B$ to be the set of length$-(r\log n)$ bitstrings which cause a false positive for $x$), it is at most:
$$\inparen{\inparen{1-\lambda}\sqrt{\frac{1}{2}} + \lambda}^{k-1}=\inparen{\inparen{1-\frac{5\sqrt{2}}{8}}\sqrt{\frac{1}{2}} + \frac{5\sqrt{2}}{8}}^{k-1}<(0.97) ^{k}$$
Which is less then $\frac{1}{n}$ when we take $c=\frac{1}{\log 0.97}$. To finish, let us note that the query complexity of $V'$ is $k\cdot O(1) = O(\log n)$ as claimed.
\end{proof}

\end{document}